\documentclass[11pt]{article}

\pdfoutput=1

\usepackage{amsmath,amsthm,amssymb,mathtools} %
\usepackage{authblk} %
\usepackage[english]{babel} %
\usepackage{cite} %
\usepackage{color} %
\usepackage{csquotes} %
\usepackage{dsfont} %
\usepackage[breaklinks]{hyperref} %
\usepackage[hmargin=1in,vmargin=1in]{geometry} %
\usepackage[normalem]{ulem}
\usepackage{graphicx} %
\usepackage{hyphenat} %
\usepackage{mathpazo} %
\usepackage{mdframed} %
\usepackage{suffix} %
\usepackage{tikz} %
\usetikzlibrary{nfold}
\usepackage{quantikz}
\usetikzlibrary{quantikz2}
\usepackage{xspace} %
\usepackage[capitalise]{cleveref}
\usepackage{nicefrac}

\usepackage{float}
\usepackage{enumitem}
\usepackage{subcaption}
\usepackage{complexity}
\usepackage{comment}
\usepackage{appendix}

\hypersetup{colorlinks=true, linkcolor=blue, citecolor=magenta}

\usetikzlibrary{calc,positioning}

\mdfdefinestyle{figstyle}{ %
  linecolor=black!7, %
  backgroundcolor=black!7, %
  innertopmargin=10pt, %
  innerleftmargin=25pt, %
  innerrightmargin=25pt, %
  innerbottommargin=10pt %
}

\definecolor{White}{rgb}{1,1,1} %
\definecolor{Black}{rgb}{0,0,0} %
\definecolor{LightGray}{rgb}{.8,.8,.8} %
\colorlet{ChannelColor}{LightGray} %
\colorlet{ChannelTextColor}{Black} %
\colorlet{ReadoutColor}{White} %

\newtheorem{theorem}{Theorem} %
\newtheorem{lemma}[theorem]{Lemma} %
\theoremstyle{definition} %
\newtheorem{definition}[theorem]{Definition} %
\theoremstyle{remark} %
\newtheorem{remark}{Remark} %

\newcommand{\norm}[1]{\ensuremath{\left\lVert #1 \right\rVert}} %
\newcommand{\abs}[1]{\ensuremath{\left\lvert #1 \right\rvert}} %

\newcommand{\complex}{\mathbb{C}} %

\newcommand{\et}{\textsc{ET}}

\newcommand{\nolines}[1]{\multicolumn{1}{c}{#1}}

\DeclareMathOperator{\tr}{Tr} %
\DeclareMathOperator{\pr}{Pr}

\begin{document}

\title{A classical proof of quantum knowledge for\\ multi-prover interactive proof systems}

\author[1]{Anne Broadbent\footnote{abroadbe@uottawa.ca}}
\author[2]{Alex B. Grilo\footnote{alex.bredariol-grilo@lip6.fr}}
\author[1]{Nagisa Hara\footnote{nhara@uottawa.ca}}
\author[1]{Arthur Mehta\footnote{amehta2@uottawa.ca}}

\affil[1]{Department of Mathematics and Statistics\protect\\
  University of Ottawa, Canada\vspace{2mm}}
\affil[2]{Sorbonne Universit\'e, CNRS, LIP6, France.\vspace{2mm}}
\date{}

\renewcommand\Affilfont{\normalsize\itshape}
\renewcommand\Authfont{\large}
\setlength{\affilsep}{6mm}
\renewcommand\Authsep{\rule{10mm}{0mm}}
\renewcommand\Authands{\rule{10mm}{0mm}}

\maketitle

\thispagestyle{empty}
\begin{abstract}
    In a  proof of knowledge (PoK), a verifier becomes convinced that a prover possesses privileged information.  In combination with zero-knowledge proof systems, PoKs play an important role in security protocols 
    such as in digital signatures and authentication schemes, as they enable a prover to demonstrate possession of certain information (such as a private key or a credential), without revealing it.
    A PoK is formally defined via the existence of an extractor, which is capable of reconstructing the key information that makes a verifier accept, given oracle access to any accepting prover.

    We extend this concept to the setting of a single classical verifier and multiple quantum provers and present the first  statistical zero-knowledge (ZK) PoK proof system for problems in~$\QMA$. To achieve this, we establish the PoK property for the ZK protocol of Broadbent, Mehta, and Zhao (TQC 2024), which applies to the local Hamiltonian problem. More specifically, we construct an extractor which, given oracle access to a provers' strategy that leads to high acceptance probability, is able to reconstruct the ground state of a local Hamiltonian. Our result can be seen as a new form of self-testing, where, in addition to certifying a pre-shared entangled state, the verifier also certifies that a prover has access to a quantum system, in particular, a ground state; this indicates a new level of verification for a proof of quantumness. 
\end{abstract} 

\newpage
\tableofcontents

\newpage

\section{Introduction}
\label{sec:intro}

Certification of quantum resources is a key task with consequences ranging from fundamental aspects of science to important practical problems. One of the most successful models for certification of quantum resources are non-local games, where two or more isolated quantum parties are tested based on a question-answer interaction with a classical verifier.

The significance of non-local games lies in their power to demonstrate that classical resources are insufficient in describing nature, such as in the experimental realization of the violation of Bell inequalities~\cite{Aspect,HBD+15}. Non-local games were also used recently in protocols of practical interest, such as device-independent cryptography~\cite{VV14b} and verification of quantum computation by classical devices~\cite{RUV13}.

These applications rely on a key property of some nonlocal games called {\em self-testing}, or rigidity. 
When this property holds for a non-local game, 
then whenever the players achieve a close-to-optimal winning probability, it follows that we can fully characterize the shared strategy. More concretely, we can prove that they must have shared a prescribed quantum state, and applied some prescribed observables during the game, up to some local change of basis. This strong level of control is a crucial step in showing the 
security of many protocols.

While self-testing techniques allow the  characterization of the entanglement shared by parties, they are not directly helpful in certifying that one of the parties  locally holds a target quantum state; what is more, these techniques necessarily allow a freedom of local isometries, rendering this type of result inconclusive in the characterization of local systems. 

Recent work has established the verification of quantum computations~\cite{Gri19,BMZ24} in the two-prover setting of a specific form: in the honest strategy, one of the parties holds a target quantum state, and teleports it to the second party, who measures it. While these works are able to conclude the \emph{existence} of a target quantum state (\emph{e.g.}~a  $\QMA$\footnote{$\QMA$ is the quantum analogue of~$\NP$~\cite{AN02arxiv}.} witness), a key question remains: 
\begin{center}
    \emph{Could the verifier ascertain that the parties actually hold such quantum state?}
\end{center}
Moreover, one could add the additional constraint that this certification be achieved without the verifier gaining any additional information, a property known as \emph{zero-knowledge} (ZK).

In this work, we answer this key question and its ZK variant, both in the positive. To this end, we revisit the cryptographic notion of {\em proof of quantum knowledge}~\cite{BG22,CVZ20} in the non-local game scenario. More concretely, we provide a formal definition of this task, and we show that a variant of the protocol~\cite{BMZ24} for verification of quantum computations  satisfies this   definition.

In order to present our results in more detail, we first overview the key concepts of zero-knowledge proofs and proofs of knowledge. 

\subsection{Proof systems: zero-knowledge and proofs of knowledge}

A classical interactive proof system is a protocol between two parties: an unbounded prover and a polynomial-time verifier. In this protocol, the prover aims to convince the verifier that, for a given~$x$, there exists a $w$ such that $(x,w) \in R$. The protocol must satisfy two key properties: completeness, which states that if $x$ is a yes-instance, the prover can convince the verifier of this fact; and soundness, which states that if $x$ is a no-instance, no prover can convince the verifier except with negligible probability. In this work, we focus on interactive proof systems where a classical verifier interacts with multiple quantum provers who are allowed to share entanglement but not to communicate during the protocol.

\emph{Zero-knowledge} (ZK) is an additional property to completeness and soundness and is of particular relevance in cryptographic scenarios: ZK captures the counter-intuitive notion that an interactive proof system can satisfy  completeness and soundness, while also completely
concealing  the proof and the  inner-workings of the prover. ZK is formalized by the existence of an efficient \emph{simulator}, who is able to emulate the view of any cheating verifier interacting with an honest prover on a positive instance. 
Here, we are interested in statistical ZK, meaning that even an unbounded distinguisher cannot discern between the two cases.
In the setting of $\MIP^*$, where we have proof systems with multiple entangled provers, \cite{GSY19,MS24} showed that every language which can be decided with an $\MIP^*$ protocol (with unbounded provers) admits a two-prover one-round zero-knowledge protocol. Moreover \cite{BMZ24} show that in the $\MIP^*$ setting, we can achieve verifiable delegation of quantum computation in zero-knowledge, extending previous results of verification of quantum computation in these models~\cite{RUV13,CGJV19,Gri19}.

A complementary property to ZK is the concept of \emph{proof of knowledge} (PoK)~\cite{GMR89,TW87,FFS88,BG93}.
This is a strengthening of the soundness condition, where the verifier not only becomes convinced of the existence of a witness but also that the prover actually ``knows" the witness.  PoK is crucial in some applications of ZK proofs, such as anonymous credentials~\cite{Cha83}: if Alice wants to authenticate herself to an online service using her private credentials using a ZK proof, we do not only want to be convinced that such credentials exist, but that Alice holds them!
More formally, a PoK for an $\NP$-relation $R$ is defined by establishing the existence of an efficient \emph{extractor}, which when given black-box access to a prover that makes the verifier accept with high enough probability, the extractor %
can efficiently compute a witness $w$ such that $(x, w) \in R$~\cite{BG93}.
We remark that for classical proof systems, PoK is typically considered together with 
ZK, since otherwise a trivial way to achieve PoK is to let the prover send~$w$ over the communication channel.

Recent works have generalized the concept of PoKs to the quantum setting, and in particular to  $\QMA$ problems. Concretely, \cite{CVZ20, BG22} define a Proof of Quantum knowledge (PoQ) in the fully quantum setting for a quantum verifier and a single quantum prover with a quantum witness. For this, they formalize the notion of a $\QMA$-relation as a natural quantum analogue to the concept of an $\NP$-relation and show that all problems in $\QMA$ admit a PoQ proof system. Furthermore, \cite{VZ21} introduces the notion of a \emph{classical} proof of quantum knowledge (cPoQ), considering PoQs where the verifier is classical. They show several examples of protocols satisfying cPoQ, and in particular they prove the existence of classical \emph{arguments} of quantum knowledge  for any $\QMA$-relation,\footnote{An argument system is similar to a proof system, but soundness is only guaranteed against bounded malicious provers.} under the quantum-hardness of LWE assumption.
We note that, in the quantum case (contrary to the classical case), PoQ can be considered independently of ZK; for example,  the trivial proof system of sending a quantum witness is disallowed in the  model with a classical verifier,  and hence establishing PoK without ZK becomes meaningful as in \cite{VZ21}. We further note that, to the best of our knowledge, the notion of 
PoQ %
has never been considered in the multi-prover setting.

\subsection{Contributions and related work}

In this work, we consider for the first time cPoQs 
to the setting of multiple non-communicating entangled provers. 
In alignment with prior work on cPoQ \cite{VZ21}, we note that key to establishing a quantum witness in the situation of classical communication only is to consider an extractor which has coherent black-box access to the prover’s unitary operation and can place a superposition of classical messages in the message register.

We highlight that our work has interesting interpretations which are unique to the multi-prover setting. Indeed, part of the motivation for this work comes from an open question in~\cite{BMZ24}, which asks for a two-prover one-round protocol that self-tests for ground states of a local Hamiltonian. We do not believe that games considered in~\cite{Gri19,BMZ24}, nor in this work can formally be considered as self-tests for the ground state. Although our extractor makes use of self-testing results as a component, it also requires the use of non-local gates in order to share teleportation keys. More generally our definition of cPoQ differs from self-testing for two reasons: (1) the extractor needs to be efficient, whereas the isometries in self-testing are not required to be efficient; and (2) the extractor is not required to act locally. In some sense, these concepts are incomparable directly with one another. Instead, we suggest viewing these two notions as capturing complementary ways of exerting control over an untrusted quantum system --- highlighting how our framework offers a unifying perspective on verification in non-local quantum settings.

Our main result is to give the first ZK-cPoQ for $\QMA$:
\begin{theorem}[Informal statement of \Cref{thm:main}]
Fix any $\QMA$ problem $A=(A_{yes} ,A_{no})$.  We give an explicit two-prover entangled proof system such that:
\begin{enumerate}
    \item \label{item:1}In the honest strategy, the provers are efficient if they are given access to a $\QMA$ witness.
    \item \label{item:2}It is statistical zero-knowledge.
    \item \label{item:3} If the provers $P^*$ with some strategy $\mathcal{S}$ make the verifier accept with sufficiently high probability, then there exists an extractor $E$ given oracle access to $\mathcal{S}$ that outputs a quantum state $\zeta$ which will be accepted with high probability by a quantum verifier for $A$.
\end{enumerate}
\end{theorem} 

We remark that our main contribution is to achieve \ref{item:3}, while maintaining both \ref{item:1} and \ref{item:2} as established in prior work.

Comparing with related work of PoQ for $\QMA$, we note that all previous protocols only achieve soundness against computationally bounded provers~\cite{VZ21,CVZ20} (i.e.~they are arguments), or achieve computational zero-knowledge~\cite{CVZ20,BG22}.

Our approach builds upon the ZK protocol for the Hamiltonian problem introduced in~\cite{BMZ24}. While the ZK property follows directly from ~\cite{BMZ24}, neither~\cite{BMZ24} nor its precursor~\cite{Gri19} (which is not ZK) can be directly used to obtain a proof of knowledge due to three key technical challenges outlined below.

The first issue arises from the fact that for the protocols from~\cite{Gri19, BMZ24}, the optimal winning probability of the given Hamiltonian games need not be close to $1$, even in the case of $\lambda_0(H)=0$. This is because in those works, the probability $p$ of playing the Energy Test versus the Pauli Braiding Test  only depends on the energy gap $\beta -\alpha$ of a given promise problem $(H, \alpha, \beta)$. This is undesirable as it would necessitate using a different knowledge error function for different problems, each of which would depend on the gap $\beta-\alpha$. In \Cref{lem:gamma} we introduce a technical modification to remedy this. The insight is that we can have more control if allow $p$ to depend on a target energy $\alpha$ as well as the average weight of the Hamiltonian terms, instead of the energy gap. This allows our protocol to gain meaningful information even when we are not in the context of a promise language.

Another more significant challenge arises from a gap in the analysis given in~\cite{Gri19,BMZ24}. In more detail, \cite{Gri19,BMZ24} construct a mapping from a local Hamiltonian $H$ to a nonlocal game $G_H$ such that: if $H$ has low ground energy, then $\omega^*(G_H) > c$; and if $H$ has high energy, then $\omega^*(G_H) < s$, for completeness and soundness parameters $c$ and $s$. However, the analysis in~\cite{Gri19,BMZ24} only bounds the winning probability of $G_H$ as a function of the ground energy of $H$, and not as a function of the energy of the specific state used by the provers during the game. While this suffices for deciding problems in $\QMA$, establishing a PoK property requires a more refined analysis. Notably, in both \cite{Gri19, BMZ24}, the optimal strategy for a given Hamiltonian $H$ need not be the ``honest'' one in which the provers share a ground state, and the precise relationship between the energy of the shared state and the acceptance probability was not explicitly analyzed. Our protocol and analysis are qualitatively stronger: we show that any strategy achieving success probability at least $1 - \epsilon$ must use a shared state with energy at most $\alpha + O(\mathrm{poly}(n))\epsilon$, where $\alpha$ denotes the ground energy of $H$, or another target energy level. 

Finally, achieving the desired result necessitates the construction of an \textit{explicit} and \textit{efficient} extractor for our protocol. In~\cite{VZ21}, an extractor for quantum money verification\footnote{As we previously mentioned, \cite{VZ21} also provide a cPoQ for \QMA{}, but the techniques used in such extractor are tailored for the protocol of \cite{Mah18} and do not use rigidity properties of non-local games.} is constructed from isometries arising from the rigidity argument, though an explicit circuit for this extractor is not provided. A similar argument can be applied to~\cite{Gri19} to obtain a cPoQ protocol, but this approach would not yield a ZK proof system.
In this work, we give the explicit circuit for our efficient extractor which makes use of the rigidity of the low-weight Pauli Braiding test used in~\cite{BMZ24}, leading to the following result.

\begin{theorem}[Informal statement of \Cref{prop:main}]
    Let $\mathcal{G}(H)$ be the Hamiltonian game  for a $XZ$ local Hamiltonian $H$. If provers $P^*$, with some strategy $\mathcal{S}$, win $\mathcal{G}(H)$ with probability $1 - \epsilon$, then the extractor with oracle access to $\mathcal{S}$ outputs a quantum state $\zeta$ such that $$\tr(H\zeta) \leq \lambda_0(H) + O(\poly(n))\,\epsilon.$$
\end{theorem}

We briefly outline the operation of the extractor (see~\Cref{fig:simple-extractor-circuit}). First, the extractor makes the first prover (Alice), who is supposed to have the ground state, teleport her state and return a teleportation key. Then, by black-box access, it also applies a ``swap" gadget to the other prover (Bob) to exchange their private states. We show that this swap gadget can be explicitly constructed by the rigidity properties of the non-local game.
Then, after the correction by the teleportation key, the extractor successfully reconstructs the target quantum state that Alice originally possessed. See \Cref{sec:construct-extractor} for further details.

\begin{figure}[h]
\centering
   \begin{quantikz}[wire types={q, q, q, q}]
        \lstick[]{$M_A$} &\gate[2]{teleport} \hphantom{wide}&&& \ctrl[vertical wire=q]{3} && \rstick{$a, b$} \\
        \lstick[2]{$\ket{\psi}_{AB}$} &&  \\
        & \gate[2]{swap} \hphantom{wide} &\\
        \lstick[]{$E$} &&&&  \gate{\sigma_X^a\sigma_Z^b} & \rstick{$\zeta$\\output state} \\
    \end{quantikz}
\caption{The simplified model of our knowledge extractor. The top wire represents the classical message register $M_A$ through which the first prover (Alice) communicates. The two middle wires are the provers' private registers, initially entangled. Then the extractor reconstructs the target state in its register on the last wire.}
\label{fig:simple-extractor-circuit}
\end{figure}

\subsection{Overview of techniques}

\subsubsection{Definitions}
\label{sec:overview_defns}

To define a PoK for quantum witnesses, we first need the definition of $\QMA$-relations, as the quantum analogue of $\NP$-relations, that quantifies the quality of the extractor's output state. Following \cite{BG22, CVZ20}, %
we fix some parameter $\gamma$ and define the relation to contain $(x, \ket{\psi})$ for all quantum states $\ket{\psi}$ that lead to acceptance probability at least $\gamma$. Therefore, fixing some quantum verifier $Q$ and $\gamma$, we define a quantum relation as follows
\[
    R_{Q,\gamma} = \{(x,\sigma) :  Q\text{ accepts } (x, \sigma) \text{ with probability at least } \gamma\}.
\]
Notice that with $R_{Q,\gamma}$, we implicitly define sets $\{\mathcal{S}_x\}_x$ such that $(x,\sigma) \in R_{Q,\gamma}$ if and only if $\sigma \in \mathcal{S}_x$.

Towards defining an extractor for our setting, we note that we face the challenge of defining an extractor for a verifier that interacts only classically with the provers---yet the extractor is required to output a quantum state, which entails that the extractor needs to interact with the provers in a quantum way.   
Hence, we adopt black-box access to a quantum prover~\cite{VZ21} to our multiple prover setting:
 the extractor is allowed to run the controlled version of the provers' observables. 
 We also allow to place any quantum state (including a coherent superposition of messages) in the public message register used to query the prover in the protocol. However, we do not allow the extractor to access the provers' private register (of unbounded size). 

\subsubsection{Proof system}

\begin{figure}
    \centering
    \begin{subfigure}[t]{0.45\textwidth}
        \centering
        \begin{tikzpicture}[>=stealth, thick]

      \node[draw, rectangle] (Pa) at (0,0) {$P_A[\phi]$};
      \node[draw, rectangle] (Pb) at (4,0) {$P_B$};
      \node[draw, rectangle] (V) at (2,-3) {$V$};

      \draw[transform canvas={xshift=-13.pt}, <-, shorten <=5pt] (Pa) -- (V) node[midway, below left, align=center] {$q_1$};
      \draw[transform canvas={xshift=-3.pt}, ->, shorten >=5pt, shorten <=3pt] (Pa) -- (V) node[midway, right, align=center] {$c_1$};
      \draw[transform canvas={xshift=3.pt}, ->, shorten <=5pt, shorten >=5pt] (Pb) -- (V) node[midway, left, align=center] {$c_2$};
      \draw[transform canvas={xshift=13.pt}, <-, shorten <=5pt] (Pb) -- (V) node[midway, below right, align=center] {$q_2$};
      \draw[double distance=3pt, nfold=3, shorten >=5pt, shorten <=5pt]  (Pa) -- (Pb) node[midway, above] {EPR};
    \end{tikzpicture}
        
        \caption{low-weight Pauli Braiding test}
    \end{subfigure}
    \hfill
    \begin{subfigure}[t]{0.45\textwidth}
        \centering
        \begin{tikzpicture}[>=stealth, thick]

      \node[draw, rectangle] (Pa) at (0,0) {$P_A[\phi]$};
      \node[draw, rectangle] (Pb) at (4,0) {$P_B$};
      \node[draw, rectangle] (V) at (2,-3) {$V$};

      \draw[transform canvas={xshift=-13.pt}, <-, shorten <=5pt] (Pa) -- (V) node[midway, below left, align=center] {``Teleport''};
      \draw[transform canvas={xshift=-3.pt}, ->, shorten >=5pt, shorten <=3pt] (Pa) -- (V) node[midway, right, align=center] {$(a, b)$};
      \draw[transform canvas={xshift=3.pt}, ->, shorten <=5pt, shorten >=5pt] (Pb) -- (V) node[midway, left, align=center] {$c_2$};
      \draw[transform canvas={xshift=13.pt}, <-, shorten <=5pt] (Pb) -- (V) node[midway, below right, align=center] {$q_2$};
      \draw[double distance=3pt, nfold=3, shorten >=5pt, shorten <=5pt]  (Pa) -- (Pb) node[midway, above] {EPR};
      
    \end{tikzpicture}
    \caption{Energy test}
    \end{subfigure}
    \caption{Hamiltonian game composed of the low-weight Pauli Braiding test and the Energy test. The Pauli Braiding test is employed to check if they share suitable many EPR pairs and perform the indicated Pauli measurements. On the other hand, during the Energy test, the verifier requests Alice to teleport a low-energy state $\phi$ to Bob and reply with the teleportation keys. Observing Bob's measurement result, the verifier evaluates the ground energy of the local Hamiltonian.}
    \label{fig:Hamiltonian-game}
\end{figure}

Our starting point is the local Hamiltonian game first introduced in \cite{Gri19}. In this game, a single classical verifier interacts with a pair of untrusted quantum provers to solve the Local Hamiltonian problem. Subsequent work \cite{BMZ24} introduced a low-weight version of the Hamiltonian game that satisfies the zero-knowledge property. In both versions, the high-level structure of the game involves provers Alice and Bob sharing a maximally entangled state. During the protocol, Alice is instructed to teleport a low-energy state of a local Hamiltonian to Bob, who then returns a measurement result to the verifier. The security of the Hamiltonian game relies on the well-known property of self-testing \cite{NV18}, which allows a classical verifier to certify that the provers share a sufficient number of EPR pairs and that Bob performs the correct measurement during the protocol.

In slightly more detail, the Hamiltonian game is composed of two sub-games referred to as the low-weight Pauli Braiding test, and the Energy test (see \Cref{fig:Hamiltonian-game}). The Pauli braiding test is a well-known self-test which is used to check if the provers share suitably many EPR pairs and perform the indicated Pauli measurements \cite{NV18}. During the Energy test, the verifier commands Alice to teleport the ground state to Bob and reply with the classical teleportation keys. The verifier then uses Bob's measurement to estimate the ground energy of the local Hamiltonian.

While self-testing is a powerful tool for analyzing the security of these games, it does not directly imply classical proofs of quantum knowledge (cPoQ). One limitation is that the formalism of self-testing is well-suited for certifying that provers share a specific quantum state—such as $n$-EPR pairs—but does not easily extend to verifying that they share an arbitrary state from a broader family, such as a low-energy state of a Hamiltonian. Moreover, this state is supposed to be only on Alice's private register, and self-testing statements concern the entanglement shared by Alice and Bob. Finally, self-testing statements imply the existence of an isometry, and in order to have an extractor, we need an {\em explicit} and {\em efficient} circuit that outputs the desired state.

\subsubsection{Constructing an extractor}\label{sec:construct-extractor}
In order to construct the extractor, we make use of recent results on the rigidity of non-local games and their links to verification of quantum computations. Our approach is to use the isometry predicted by rigidity to reconstruct an accepting state. 

\begin{figure}
    \centering
    \begin{subfigure}[t]{0.45\textwidth}
        \centering
        \begin{tikzpicture}[>=stealth, thick]

          \node[draw, rectangle] (Pa) at (0,0) {$P_A[\cdot]$};
          \node[draw, rectangle] (Pb) at (4,0) {$P_B[\phi]$};
          \node[draw, rectangle] (E) at (2,-3) {$E$};
        
          \draw[transform canvas={xshift=-13.pt}, <-, shorten <=5pt] (Pa) -- (E) node[midway, below left, align=center] {1.Teleport};
          \draw[transform canvas={xshift=-3.pt}, ->, shorten >=5pt, shorten <=3pt] (Pa) -- (E) node[midway, right, align=center] {$(a,b)$};
          \draw[decorate, decoration={snake}, ->,]  (0.8, 0) -- (3.2, 0) node[midway, above] {teleported};
        \end{tikzpicture}
        \caption{Alice teleports the ground state to Bob}
    \end{subfigure}
    \hfill
    \begin{subfigure}[t]{0.45\textwidth}
        \centering
        \begin{tikzpicture}[>=stealth, thick]

          \node[draw, rectangle] (Pa) at (0,0) {$P_A$};
          \node[draw, rectangle] (Pb) at (4,0) {$P_B[\cdot]$};
          \node[draw, rectangle] (E) at (2,-3) {$E[\phi]$};
        
          \draw[transform canvas={xshift=-13.pt}, <-, shorten <=5pt, dotted] (Pa) -- (E) node[midway, below left, align=center] {};
          \draw[transform canvas={xshift=-3.pt}, ->, shorten >=5pt, shorten <=3pt, dotted] (Pa) -- (E) node[midway, right, align=center] {};
          \draw[decorate, decoration={snake}, ->, dotted]  (0.8, 0) -- (3.2, 0);
          \draw[shorten >=5pt, shorten <=5pt, ->, ultra thick] (Pb) -- (E) node[midway, below right, align=center] {2.Swap \\ (Black-box access)};
        \end{tikzpicture}
        \caption{Bob and the extractor swap their states}
    \end{subfigure}
    \caption{The process of extracting the ground state from Alice. First, the extractor makes Alice teleport the ground state to Bob as she does in the Energy test. Then the extractor swaps its internal state with Bob's, which is achieved by the extractor's black-box access. After the correction using teleportation keys, the extractor finally obtains the desired state.}
    \label{fig:extracting-process}
\end{figure}

More concretely, based on the assumption that they succeed at the Hamiltonian game with large probability, we construct an extractor as follows (see \Cref{fig:extracting-process}): First, the extractor sends Alice a query to teleport her witness $\phi$. Consequently, such witness is teleported into Bob's register. Next, the extractor exchanges its internal state with Bob's private register by black-box access to Bob's observables, which functions as a ``swap gadget'' and enables the extractor to retrieve the witness from Bob's hand. We note that both operations are performed locally on each prover’s register. Hence, these two steps commute and they can be carried out interchangeably (or simultaneously). Then one can also interpret the extractor in the way that it first exchanges its internal state with Bob's state to make itself entangled with Alice, and it next orders Alice to teleport her witness to receive it.

A key feature of the previous paragraph is to apply a ``swap gadget'' of the isometry. The existence of such an isometry is proven in the context of rigidity by the Gowers-Hatami theorem~\cite{GH17}, which is also used as a tool to show rigidity in the protocol~\cite{VZ21}. However, by close observation of each isometry discussed in \cite{VZ21, BMZ24}, it turns out that the isometry essentially emulates the swap gate. This is why we call the isometry a ``swap gadget''.

\subsection{Further related work and open problems. }
We now further discuss how our results compare to related results and state some open problems.

\paragraph{cPoQ vs.~verification.} As mentioned before, previous results~\cite{Gri19,BMZ24} have considered the task of verification of quantum computation. The problem they solve is roughly the following: given a Hamiltonian $H$, is there a state with energy below the threshold $\alpha$, or all states have energy above the threshold $\beta$? In their proof, they show that if the provers pass the test, then such a state must exist. We take an important step further and show that not only such a state exists, but the provers actually have such a state in hand. 

Moreover, we also notice that our statement is still useful when we know a priori that such a state exists, whereas the question of the existence of such a state trivializes in this case, and this could be important in some applications. For example, let us suppose that a company claims that they have the technology to create a known low-energy state of a Hamiltonian $H$. Using our results, one can certify that this claim is true. Remarkably, the cPoQ property  is a thought experiment, and in our case, it does not require additional resources beyond the conventional proof system: by theoretically demonstrating the existence of the extractor, we have upgraded the conclusions that can be drawn from a successful run of the initial protocol.

Finally, we notice that other protocols for verification of quantum computation with entangled servers~\cite{RUV13,GKW15,CGJV19} verify the computation step-by-step, instead of using the circuit-to-Hamiltonian construction. We leave it as an open problem if one can prove proof of quantum knowledge for this type of protocols as well.

\paragraph{Proofs of quantum knowledge in other models.}
Proofs of quantum knowledge have been considered in other settings such as \cite{BG22,CVZ20,VZ21}. In \cite{BG22}, to achieve proof of quantum knowledge, they have a protocol where a single prover exchanges messages with a verifier with quantum capabilities. In~\cite{CVZ20}, they remove the need of interaction, at a cost of requiring a trusted quantum setup shared between the prover and the verifier. Both of these results have the advantage of requiring a single server, and therefore, no need for the space-like separation and the challenge of keeping highly entangled states between two quantum devices. However, we think that achieving proof of quantum knowledge with classical devices is an important feature for fundamental reasons and practical ones. First, this provides a stronger classical ``leash'' on quantum devices. Secondly, in the two-prover scenario, quantum communication is only needed by previously established parties, and in a later moment anyone that wants to verify their resources can be fully classical.

This is also the motivation of \cite{VZ21}, where they provide a classical proof of quantum knowledge under cryptographic assumptions. In their result, they extend the  protocol of classical verification of quantum computation~\cite{Mah18b}, and show a classical proof of quantum knowledge against a {\em bounded }quantum prover. While again, their protocol has the advantage of holding in the single-server scenario, our cPoQ has some advantages when compared to theirs. First, it is information-theoretically secure, and theirs relies on the hardness of LWE. More importantly, their protocol may require much more powerful devices than the ones needed to create the target quantum state.
They also assume that the prover prepares a polynomial number of witness copies for the protocol, whereas our model requires only a single copy.
We thus conclude that our approach can lead to small-scale demonstrations even in the short term, thus boosting the conclusions of quantum experiments such as~\cite{DNM+24}.
Furthermore, \cite{VZ21} suggests but does not formally study the ZK property, whereas we provide a full analysis also for ZK.  

\paragraph{State complexity.} We notice that recently quantum complexity classes have been extended to ``inherently quantum problems''~\cite{RY22,MY23,BEM+23arxiv}. In particular, new classes such as \class{stateBQP}, \class{stateQMA}, etc., have been proposed, where the goal is to synthesize a quantum state, and not solve a decision problem or find a classical solution as in  standard complexity classes. We leave as an open direction to study cPoQs as verifiable delegation of state synthesis problem.

\paragraph{Extension of self-testing technique.} As we previously discussed, our definition of cPoQ does not follow the traditional definition of self-testing. On one hand, we allow more power in the extraction which is not required to use local isometries. On the other hand the extractor is required to be efficient, a restriction not required in the formal definition of self-testing.  We leave it as an open problem if one can use this alternative model of certification in other scenarios.

\paragraph{Proofs of destruction.} 
In quantum cryptography, there is a recent line of works that consider proofs of destruction of states that were held by a second party~\cite{BI20}.
We notice that our extractor gets their hands on a quantum state that was originally held by one of the provers. We leave as an open question if our definition can be used to prove that, under some assumptions, one of the provers ``lost'' their knowledge of the witness.

\paragraph{Compiling non-local games.}  We leave as an open question whether we could also obtain cPoQ in the complied interactive single-prover protocol following the lines of~\cite{KLVY23}.

\subsection{Acknowledgements}
We thank Yuming Zhao for helpful discussions.
The authors acknowledge support of the Institut Henri Poincaré (UAR 839 CNRS-Sorbonne Université), and LabEx CARMIN (ANR-10-LABX-59-01).
We acknowledge the support of the Natural Sciences and Engineering Research Council of Canada (NSERC)(ALLRP-578455-2022). This research was undertaken, in part, thanks to funding from the Canada Research Chairs Program. ABG is funded by ANR JCJC TCS-NISQ ANR-
22-CE47-0004 and by the PEPR integrated project EPiQ ANR-22-PETQ-0007 part of Plan France
2030. AM is supported by NSERC DG 2024-06049.

\section{Preliminaries}\label{sec:prelims}
\subsection{Notation}

Throughout this paper, we regard any $n$-bit strings as an element of $\mathbb{Z}_2^{\otimes n}$. 
For any $n$-bit strings $a, b$ the notation $a + b$ denotes the bit-wise XOR of $a$ and $b$. We denote by $|a|$ the Hamming \mbox{weight of~$a$.} 

In this work, all Hilbert spaces $\mathcal{H}$ are taken to be finite dimensional and we use $\| \cdot \|$ to denote the usual $\ell^2$ norm on $\mathcal{H}$.

\subsection{Quantum information}

We follow the standard quantum formalism of states and measurements. An
\emph{observable} is a Hermitian operator whose eigenvalues are $\pm
1$, and encodes a two-outcome projective measurement (the POVM elements of the two outcomes
are the projections on to the $+1$ and $-1$ eigenspaces).
In particular, we use the following Pauli matrices and the identity matrix denoted by
\[
    \sigma_I \coloneqq \begin{pmatrix} 1 & 0 \\ 0 & 1 \end{pmatrix},
    \sigma_X \coloneqq \begin{pmatrix} 0 & 1 \\ 1 & 0 \end{pmatrix}, 
    \sigma_Z \coloneqq \begin{pmatrix} 1 & 0 \\ 0 & -1 \end{pmatrix}.
\]
They satisfy the anti-commutation relation
\[ 
    \sigma_X\sigma_Z = -\sigma_Z\sigma_X.
\]
Let $\ket{\Phi^+}$ to denote the EPR pair, that is, 
\[
    \ket{\Phi^+} = \frac{1}{\sqrt{2}}(\ket{00} + \ket{11})
\]and $\ket{\Phi^+_n}$ to denote the $n$-tensor product of the EPR pair.

Consider two quantum systems $A$ and $B$ whose Hilbert spaces are given $H_A$ and $H_B$ respectively.
We denote by $\tr_A$ the partial trace over $A$. We also define a function $\tr_{\overline{A}}$ to obtain a reduced density operator on $A$. For example, let $\rho_{AB}$ be a density operator defined on $H_A \otimes H_B$. Then $\tr_{\overline{A}}(\rho_{AB})$ is a reduced density operator on $H_A$.

Finally, we recall the notion of efficient uniform family of quantum circuits.

\begin{definition}\label{def:polytime-extractor}
    A quantum polynomial-time machine $\mathcal{C}$ is a uniformly generated family of quantum circuits $\mathcal{C} = \{\mathcal{C}_n\}_{n}$, where, for some polynomials $p, q, r$, $\mathcal{C}_n$ takes as input a string $z \in \Sigma^*$ with $|z| = n$, a $p(n)$-qubit quantum state $\ket{\phi}$, and $q(n)$ auxiliary qubits in state $\ket{0}^{\otimes q(n)}$, consists of $r(n)$ gates.
In particular, we call it an efficient machine.
\end{definition}

\subsection{Non-local games}\label{sec:nonlocalgame}
A two-player (called Alice and Bob) one-round non-local game $\mathcal{G}$ is a tuple $\big(\lambda,\mu,\mathcal{I}_A,\mathcal{I}_B,\mathcal{O}_A,\mathcal{O}_B\big)$ , where $\mathcal{I}_A,\mathcal{I}_B$ are finite input sets, and $\mathcal{O}_A,\mathcal{O}_B$ are finite output sets, $\mu$ is a probability distribution on $\mathcal{I}_A\times\mathcal{I}_B$, and $\lambda:\mathcal{O}_A\times\mathcal{O}_B\times\mathcal{I}_A\times\mathcal{I}_B\rightarrow\{0,1\}$ determines the win/lose conditions. A quantum strategy $\mathcal{S}$ for~$\mathcal{G}$
is given by finite-dimensional Hilbert spaces $\mathcal{H}_A$ and $\mathcal{H}_B$, a unit vector $\ket{\psi}\in\mathcal{H}_A\otimes \mathcal{H}_B$, Alice's POVMs $\{E_a^x:a\in\mathcal{O}_A \}, x\in\mathcal{I}_A$ on $\mathcal{H}_A$, and Bob's POVMs $\{F_b^y:b\in\mathcal{O}_B\},y\in\mathcal{I}_B$ on $\mathcal{H}_B$. The winning probability of $\mathcal{S}$ for game $\mathcal{G}$ is given by
\begin{align*}
\omega(\mathcal{G}, \mathcal{S}):=\sum\limits_{a,b,x,y}\mu(x,y)\lambda(a,b\vert x,y)\bra{\psi}E_a^x\otimes F_b^y\ket{\psi}.
\end{align*}
A quantum strategy $\mathcal{S}$ for a non-local game $\mathcal{G}$ is said to be \emph{perfect} if $\omega(\mathcal{G}, \mathcal{S})=1$. When the game is clear from the context we simply write $\omega(\mathcal{S})$ to refer to the winning probability of a strategy $\mathcal{S}$. The \emph{quantum value} of a non-local game $\mathcal{G}$ is defined as
 \begin{align*}
     \omega^*(\mathcal{G}):=\sup\{\omega(\mathcal{S}):\mathcal{S} \mbox{  a quantum strategy for } G\}.
 \end{align*}

In this paper, we assume all measurements employed in a quantum strategy are PVMs. An $m$-outcome PVM $\{P_1,\cdots,P_m\}$ corresponds to an observable $\sum_{j\in[m]}\exp(\frac{2\pi i}{m}j)P_j$, so a quantum strategy for a game $\mathcal{G}=\big(\lambda,\mu,\mathcal{I}_A,\mathcal{I}_B,\mathcal{O}_A,\mathcal{O}_B\big)$ can also be specified by a triple
\begin{align*}
    \mathcal{S}=(\tau^A,\tau^B,\ket{\psi}\in\mathcal{H}_A\otimes\mathcal{H}_B)
\end{align*}
where $\tau^A(x)$, $x\in \mathcal{I}_A$ are $\mathcal{O}_A$-outcome observables on $\mathcal{H}_A$, and $\tau^B(y)$, $y\in \mathcal{I}_B$ are $\mathcal{O}_B$-outcome observables on $\mathcal{H}_B$.

Here we introduce the well-known Mermin-Peres Magic Square game~\cite{Mer90,Per90}, in which Alice and Bob are trying to convince the verifier that they have a solution to a system of equations over $\mathbb{Z}_2$. There are 9 variables $v_1, \dots, v_9$ in a $3\times 3$-array whose rows are labeled $r_1,r_2,r_3$ and columns are labeled $c_1, c_2, c_3$.
\begin{table}[ht]
    \centering
    \begin{tabular}{c|c|c|c|}
    \nolines{} & \nolines{$c_1$} & \nolines{$c_2$} & \nolines{$c_3$}  \\ \cline{2-4}
    $r_1$ & $v_1$ & $v_2$ & $v_3$  \\ \cline{2-4}
    $r_2$ & $v_4$ & $v_5$ & $v_6$ \\ \cline{2-4}
    $r_3$ & $v_7$ & $v_8$ & $v_9$ \\ \cline{2-4}
    \end{tabular}
    \caption{Magic square game}
\end{table}

Each row or column corresponds to an equation: variables along the rows or columns in $\{r_1,r_2,r_3,c_1,c_2\}$ sum to 0; variables along the column $c_3$ sum to 1. In each round, Alice receives one of the 6 possible equations and  must respond with a satisfying assignment to the given equation. Bob is then asked to provide a consistent assignment to one of the variables contained in the equation Alice received. The following table describes an operator solution for this system of equations:
\begin{table}[h]
    \centering
    \begin{tabular}{lll}
        $A_1= \sigma_I \otimes \sigma_{Z}$ &  $A_2 = \sigma_{Z} \otimes \sigma_I$ & $A_3 = \sigma_{Z} \otimes \sigma_{Z}$ \\
         $A_4 = \sigma_{X} \otimes \sigma_I$ & $A_5 = \sigma_I \otimes \sigma_{X}$ & $A_6 = \sigma_{X} \otimes \sigma_{X}$ \\
         $A_7 = \sigma_{X} \otimes \sigma_{Z}$ & $A_8 = \sigma_{Z} \otimes \sigma_{X}$ & $A_9 = \sigma_X\sigma_Z \otimes \sigma_Z\sigma_X$
    \end{tabular}
\caption{Operator solution for Magic Square game}\label{tbl:optStrategyforMSgame}
\end{table}

This game can be won with certainty by the following strategy $\mathcal{S^*}$:
\begin{itemize}
    \item the players share two EPR pairs,
    \item given a variable $v_i$, Bob performs $A_i$ on his registers, and
    \item given a row or column consisting of three variables $v_j,v_k$ and $v_\ell$, Alice perform $A_jA_kA_\ell$ on her registers.
\end{itemize}

\begin{definition}
Let $\mathcal{S}=(\tau^A,\tau^B,\ket{\psi}\in\mathcal{H}_A\otimes\mathcal{H}_B)$ and $\widetilde{\mathcal{S}}=( \{\widetilde{\tau}^A \},\{\widetilde{\tau}^B \},\ket{\widetilde{\psi}} \in \widetilde{\mathcal{H}}_A \otimes \widetilde{\mathcal{H}}_B )$ be two quantum strategies for a game $\mathcal{G}=\big(\lambda,\mu,\mathcal{I}_A,\mathcal{I}_B,\mathcal{O}_A,\mathcal{O}_B\big)$. We say  $\mathcal{S}$ is $\delta$-close to  $\widetilde{\mathcal{S}}$ 
if there are Hilbert spaces $\mathcal{H}_A^{aux}$ and $\mathcal{H}_B^{aux}$,  isometries $V_A:\mathcal{H}_A\rightarrow \widetilde{\mathcal{H}}_A\otimes\mathcal{H}_A^{aux}$ and $V_B:\mathcal{H}_B\rightarrow\widetilde{\mathcal{H}}_B\otimes\mathcal{H}_B^{aux}$, and a unit vector $\ket{aux}\in\mathcal{H}_A^{aux}\otimes\mathcal{H}_B^{aux}$ such that
\begin{align}
    \norm{ (V_A\otimes V_B)(\tau^A(x)\otimes \tau^B(y)\ket{\psi}) - (\widetilde{\tau}^A(x)\otimes\widetilde{\tau}^B(y)\ket{\widetilde{\psi}})\otimes\ket{aux} }^2\leq\delta
\end{align}
for all $(x,y)\in\mathcal{I}_A\times\mathcal{I}_B$.
\end{definition}

The rigidity of the Magic Square game has been well studied \cite{WBMS16}:
\begin{lemma}\label{lemma:MS}
If $\mathcal{S}$ is a strategy for the Magic Square game with winning probability $1-\epsilon$, then $\mathcal{S}$ is $O(\sqrt{\epsilon})$-close to $\mathcal{S^*}$.
\end{lemma}

\subsection{Local Hamiltonians}
We define the Local Hamiltonian problem known as the quantum analog of MAX-SAT problem. In a nutshell, the Local Hamiltonian problem asks if there is a global state such that its energy in respect of $H = (1/m)\sum_{i \in [m]} H_i$ is at most $\alpha$ or all states have energy at least $\beta$ for some constant $\alpha, \beta \in \mathbb{R}$ with $\alpha < \beta$. This problem was first proved to be QMA-complete for $k = 5$ and $\beta - \alpha \geq 1/\poly(n)$ \cite{ksv02}. 
In this paper, we particularly consider $XZ$ local Hamiltonians~\cite{CM14} where all the terms are tensor products of $\sigma_X, \sigma_Z$ and $\sigma_I$.
\begin{definition}[XZ Local Hamiltonian]\label{def:xz-local-hamiltonian}
    The XZ $k$-Local Hamiltonian problem, for $k \in \mathbb{N}$ and parameters $\alpha, \beta \in [0,1]$
    with $\alpha<\beta$, is the following promise problem. Let $n$ be the number of qubits of a quantum system.
    The input is a sequence of $m(n)$
    values $\gamma_1,...,\gamma_{m(n)}\in [-1,1]$ and $m(n)$
    Hamiltonians $H_1, \ldots, H_{m(n)}$
    where $m$ is a polynomial in $n$, and for each $i\in [m(n)]$,
    $H_i$ is of the form $\bigotimes_{j \in n} \sigma_{W_j} \in  \{\sigma_X, \sigma_Z, \sigma_I\}^{\otimes n}$ with
    $|\{j | j \in [n] \textrm{ and } \sigma_{W_j} \neq \sigma_I \}| \leq k$.
    For $H = \frac{1}{m(n)} \sum_{j = 1}^{m(n)} \gamma_j H_j$, one of the following two conditions hold.
    \begin{description}
        \item[\quad Yes.]
        {There exists a
          state $\ket{\psi} \in \complex^{2^{n}}$ such that
          $\bra{\psi} H \ket{\psi}
            \leq \alpha(n)$}
        \item[\quad No.]
        {For all states $\ket{\psi} \in \complex^{2^{n}}$
          it holds that
          $\bra{\psi} H \ket{\psi}
            \geq \beta(n) .$
            }
    \end{description} 
\end{definition}

\subsection{Entangled proof systems and zero-knowledge}

Here give the definition for a one-round two-prover entangled proof system to be statistical zero-knowledge.

\begin{definition}
    A one-round two-prover entangled proof system for a problem $A=(A_{yes}, A_{no})$, with completeness $c$ and soundness $s$, is a polynomial-time computable function that takes an instance $x \in A$ to a description of a non-local game $\mathcal{G}_x$ 
satisfies the following conditions. 
\ \begin{description}
        \item[\quad Completeness.]
        {For every
          $x \in A_{yes}$ we have
          $\omega^*(\mathcal{G}_x) \geq c$}
        \item[\quad Soundness.]
        {For every
          $x \in A_{no}$ we have
          $\omega^*(\mathcal{G}_x) \leq s$.}
    \end{description} 
\end{definition}

Zero-knowledge is defined with respect to a prescribed honest strategy. Specifically, for the honest quantum strategy $\mathcal{S}_h$ of Alice and Bob and malicious verifier $\widehat{V}$, we take $View(\widehat{V}(x), \mathcal{S}_h)$ to be the random variable corresponding to the transcript $(x,\theta, q_1,r_1,q_2,r_1)$ where $x$ and $\theta$ are the input and randomness of $\widehat{V}$, and $q_1, q_2$ and $r_1, r_2$ are questions ans answers of Alice and Bob respectively. A protocol is zero-knowledge when for all ``yes" instances, a simulator can sample from the distribution above. 

\begin{definition}\label{defn:SZKMIP*}
    A one-round two-prover entangled proof system for a problem $A=(A_{yes}, A_{no})$ is \emph{statistical zero-knowledge} if for every $x \in A_{yes}$ there exists an honest prover strategy $\mathcal{S}_h$ satisfying the following.

    \begin{description}
        \item[\quad Completeness.]
        {Provers using $\mathcal{S}_h$ win the game with probability at least $c$, i.e
          $\omega^*(\mathcal{G}_x, \mathcal{S}_h) \geq c$.}
        \item[\quad Simulability.]
        {For any PPT (probabilistic polynomial-time) malicious verifier $\widehat{V}$ there exists a PPT simulator $Sim_{\widehat{V}}$ with output distribution that is $\epsilon$-close to $View(\widehat{V}(x),\mathcal{S}_h)$ in statistical distance for some negligible function $\epsilon(|x|)$.}
    \end{description}

\end{definition}

\subsection{Classical proofs of knowledge and $\mathsf{QMA}$ relations}

We refer the reader to \cref{sec:intro} for a high-level introduction to Proofs of Knowledge (Pok).  Here, we provide further background details and definitions.   

A PoK is an interactive proof system for some relation $R$ such that if the verifier accepts some input $x$ with high enough probability, then she is convinced that the prover ``knows'' some witness $w$ such that $(x, w) \in R$. This notion is formalized by requiring the existence of an efficient extractor~$K$ that is able to return a witness for $x$ when $K$ is given oracle access to the prover. 
Here, oracle access means that the machine can operate the prover's action, and rewind the prover, i.e., after some interaction with the prover it can revert the prover back to a previous state and retry another interaction observing the response.

\begin{definition}[classical PoK \cite{BG93}]
    Let $R \subseteq \mathcal{X} \times \mathcal{Y}$ be a relation. A proof system $(P, V )$ for $R$ is a PoK for $R$ with knowledge error $\kappa$ if there exists a polynomial $p > 0$ and a polynomial-time machine $K$ , called the knowledge extractor,
    such that for any classical interactive machine $P^*$ that makes $V$ accept some instance $x$ of size $n$ with probability at least $\varepsilon > \kappa (n)$, we have \[
    Pr \Bigl[ \Bigl( x, K^{P^*} (x,y)(x) \Bigr) \in R \Bigr] \geq p \biggl( (\varepsilon  - \kappa (n)), 1^n \biggr).
\]
In the definition, $y$ corresponds to the side-information that $P^*$ has, possibly including some $w$ such that $(x, w) \in R$.
\end{definition}

Then this notion was extended in the post-quantum setting by Unruh \cite{Unr12}, followed by \cite{CVZ20} and \cite{BG22} in the fully quantum setting.
However, we see two major challenges in the quantum setting: First, it is not straightforward to define the quantum relation between the input $x$ and some quantum state $\ket{\psi}$. Second, how the extractor interacts with the prover must be clarified. In particular, unlike in a classical setting, one cannot utilize the rewinding technique (taking a snapshot of the state and rewinding to the previous state) due to the no-cloning theorem and the destructive nature of quantum measurements.

To overcome the first challenge, we introduce a quantum relation \cite{BG22}. We fix some parameter $\gamma$ and define the relation to contain $(x, \ket{\psi})$ for all quantum states $\ket{\psi}$ that lead to acceptance probability at least $\gamma$. Therefore, fixing some quantum verifier $Q$ and $\gamma$, we define a quantum relation as follows
\[
    R_{Q,\gamma} = \{(x,\sigma) :  Q\text{ accepts } (x, \sigma) \text{ with probability at least } \gamma\}.
\]
Notice that with $R_{Q,\gamma}$, we implicitly define sets of states $\{\mathcal{S}_x\}_x$ such that $(x,\sigma) \in R_{Q,\gamma}$ if and only if $\sigma \in \mathcal{S}_x$. With this in hand, we can define a $\QMA$-relation.

\begin{definition}[$\QMA$-relation] 
    Let $A = (A_{yes}, A_{no})$ be a problem in $\QMA$, and let $Q$ be an associated quantum polynomial-time verification algorithm (which takes as input an instance and a witness), with completeness $\alpha$ and soundness $\beta$. Then, we say that ($R_{Q,\gamma}, \alpha, \beta)$  is a $\QMA$-relation with completeness $\alpha$ and soundness $\beta$ for the problem $A$ if for all $x \in A_{yes}$, there exists some $| \psi \rangle$ such that $(x, | \psi \rangle ) \in R_{Q,\alpha}$ and for all $x \in A_{no}$, for every $\rho$ it holds that $(x, \rho ) \not \in R_{Q,\beta}$.
\end{definition}

\section{Proof of quantum knowledge for multi-prover interactive proof systems}
\label{sec:cpok-multiple-provers}
In this section, we formally define oracle access to the multiple provers with some strategy, as well as classical proofs of quantum knowledge for multiple-prover interactive proof systems. We refer the reader to \cref{sec:intro} for a high-level introduction.

As mentioned earlier, we extend the definition of an extractor made in \cite{VZ21}. The extractor has black box access to the provers,
being able to coherently manipulate the prover's unitary operation controlled on a messages from the verifier. For example, the extractor is able to put a superposition of messages from the verifier in the message register~\cite{Wat09b,Unr12}, and the provers coherently apply their operations on their quantum state.

In order to formally define cPoQs, we first model a machine as a quantum circuit that has a oracle access to multiple provers. 
\begin{definition}[Oracle access to a strategy]\label{def:OracleAccess}
Let $\mathcal{S}$ denote a strategy for a one-round two-prover entangled proof system $\mathcal{G}$. 
We denote by $E^\mathcal{S}$ a quantum circuit $E$ that has a black-box access to the provers consisting of its private register as well as prover registers initialized in state $\ket{\psi}$.
The circuit consists of arbitrary gates from a universal family applied on $\mathcal{H}_E$, together with controlled applications of prover observables on each provers' registers.    
\end{definition}

We notice that by {\em applying} an observable, we mean that the extractor can either apply the operator corresponding to the observable, or measure the quantum state and report the outcome.

Finally, we define a classical Proof of Quantum Knowledge for multiple-prover interactive proofs.

\begin{definition}[Classical Proof of Quantum Knowledge for multiple-prover interactive proofs]
  \label{def:cpoq-multi}
  Let $R_{Q,q(\epsilon, n)}$ be a $\QMA$-relation.
  A one-round two-prover entangled proof system $\mathcal{G}$ is a classical Proof of Quantum Knowledge for $R_{Q,q(\epsilon, n)}$ with knowledge error $\kappa(\cdot) > 0$ and quality function $q$,
  if there exists a polynomial-time machine $E$ which given oracle access to  any strategy $\mathcal{S}$ satisfying  $\omega^*(\mathcal{G},\mathcal{S}) \geq  \epsilon > \kappa(x, 1^n)$, outputs a quantum state $\phi$ in its private register such that $\left(x, \phi \right) \in R_{Q,q(\epsilon,n)}$.
\end{definition}

\section{Non-local game for local Hamiltonians}\label{sec:nonlocalgame-localhamiltonians}
In~\cite{Gri19}, Grilo introduces a non-local game for deciding the XZ-Hamiltonian problem, as given in \cref{def:xz-local-hamiltonian}. Honest provers for this game share suitably many EPR pairs, and one prover (Alice) holds a ground state for Hamiltonian~$H$. 
The game combines two subtests: an energy test, and a second test used to detect deviations from honest behavior. In the energy test, the verifier asks Alice to teleport the ground state of a local Hamiltonian $H$ to Bob through the shared EPR pairs. Then Bob performs a measurement corresponding to a randomly chosen Hamiltonian term. For the second test, ~\cite{Gri19} uses the well-known Pauli Braiding test~\cite{NV17} to certify that the players share suitably many EPR pairs, and perform the correct measurements. 

Subsequently, ~\cite{BMZ24} presented a modified version of this game in which the Pauli Braiding test is replaced by the low-weight Pauli Braiding test. This modification was motivated by improving the game of ~\cite{Gri19} to a zero-knowledge protocol. We chose to work with game introduced in ~\cite{BMZ24} to allow for the construction of a zero-knowledge classical proof of quantum knowledge.

\subsection{Low-weight Pauli braiding test}

First we recall a few details about the LWPBT~\cite{BMZ24}, which is constructed from the low-weight linearity test and the low-weight anti-commutation test (see \Cref{fig:pbt}). In brief, during this game Bob receives questions from the set $\mathcal{I}_B:=\{W(a):W\in\{X,Z\}^n,a\in\{0,1\}^n \text{ such that }\abs{a}\leq 6\}$, whereas Alice receives similarly formed questions in pairs. In the honest strategy, the provers share $n$ EPR pairs and measure them according to the appropriate Pauli observable determined by the questions $W(a)$.

\begin{figure}[t]
\rule[1ex]{\textwidth}{0.5pt}
The verifier performs the following steps, with probability $\frac{1}{2}$ each:
  \begin{enumerate}[label=(\Alph*)]
  \item Linearity test
    \begin{enumerate}[label=\arabic*.]
      \item  The verifier selects uniformly at random $W \in \lbrace X, Z \rbrace^{n}$  and strings $a, a' \in \lbrace 0,1 \rbrace^{n}$ satisfying $|a|, |a'| \leq 6$ (\emph{i.e.} $a,a'$ both have at most $6$ non-zero entries).
      \item The verifier sends $(W(a), W(a'))$ to Alice. If $a+a'$ has weight at most $6$ then the verifier selects $W' \in \lbrace W(a), W(a'), W(a+a') \rbrace$ uniformly at random to send to Bob. Otherwise, the verifier uniformly at random sends $W' \in \lbrace W(a), W(a') \rbrace$ to Bob.
      \item The verifier receives two bits $(b_1,b_2)$ from Alice and one bit $c$ from Bob.
      \item If Bob receives $W(a) $ then the verifier requires $b_1=c$. If Bob receives $W(a')$ then the verifier requires $b_2=c$. If Bob receives $W(a+a')$ then the verifier requires $b_1 +b_2 =c$.
    \end{enumerate}
  \item Anti-commutation test
    \begin{enumerate}[label=\arabic*.]
      \item The verifier samples uniformly at random a string $a \in \{0, 1\}^n$ with exactly two non-zero entries $i < j$. The verifier also samples a row or column $q \in \{r_1, r_2, r_3, c_1, c_2, c_3\}$, and a variable $v_k$ contained in $q$ as in the Magic Square game.
      \item Alice receives the question $(q, a)$.
      \item If $k \neq 9$ then Bob receives $W (a) = I^{i-1}W^{i}I^{j-i}W^{j}I^{n-j} \in \mathcal{I}_A$ with $\sigma_{W_i} \otimes \sigma_{W_j} = A_k$. If $k = 9$ then Bob receives question $(v_9, a)$.
      \item The players win if and only if Alice responds with a satisfying assignment to $q$ and Bob provides an assignment to variable $v_k$ that is consistent with Alice’s.
    \end{enumerate}
\end{enumerate}
\rule[2ex]{\textwidth}{0.5pt}\vspace{-.5cm}
\caption{Low-Weight Pauli Braiding Test}\label{fig:pbt}
\end{figure}

While the honest strategy wins with probability $1$, in general, the provers could deviate and perform an arbitrary strategy $\mathcal{S}=(\tau^A, \tau^B, \ket{\psi}_{AB} \in \mathcal{H}_A \otimes \mathcal{H}_B)$, sharing some entangled state $\ket{\psi}_{AB} \in \mathcal{H}_A \otimes \mathcal{H}_B$. We will write $W_P(a) = \tau^P(W(a)), P \in \{A, B\}$ to denote prover's observable corresponding to questions $W(a) \in \mathcal{I}_B$. The self-testing properties of this game show that if a strategy $\mathcal{S}$ is playing suitably close to $1$ then the state and measurements must be, up to local isometry, close to the honest strategy. 

\begin{theorem}[Theorem 18 of \cite{BMZ24}]\label{thm:rigidity}
There exists a constant $C>0$ such that the following holds. For any $\epsilon>0$, $n\in\mathbb{N}$, and strategy $\mathcal{S}=(\tau^A,\tau^B,\ket{\psi}_{AB}\in\mathcal{H}_A\otimes\mathcal{H}_B)$ for the $n$-qubit LWPBT with winning probability $1-\epsilon$, there are isometries $V_A:\mathcal{H}_A\rightarrow (\mathbb{C}^2)^{\otimes n}\otimes\mathcal{H}_A^{aux},V_B:\mathcal{H}_B\rightarrow(\mathbb{C}^2)^{\otimes n}\otimes\mathcal{H}_B^{aux}$ and a unit vector $\ket{aux}\in\mathcal{H}_A^{aux}\otimes\mathcal{H}_B^{aux}$ such that
\begin{equation*}
   \norm{ (V_A\otimes V_B)\big(Id_{\mathcal{H}_A} \otimes W_B(a)\ket{\psi}_{AB}\big)-\big(Id_{\mathbb{C}^{2^n}} \otimes \sigma_W(a)\ket{\Phi^+}^{\otimes n}\big)\otimes\ket{aux} } \leq Cn^{6}\epsilon^{1/4}
\end{equation*}
for all $W(a)\in\mathcal{I}_B$.

In particular, we have
\begin{equation*}
   \norm{W_B(a)\ket{\psi}_{AB}-V_B^\dagger \sigma_W(a)V_B\ket{\psi}_{AB}} \leq Cn^{6}\epsilon^{1/4}
\end{equation*}
for all $W(a)\in\mathcal{I}_B$.
\end{theorem}

In order to construct our extractor we require an explicit form for the isometry from \cref{thm:rigidity}.  It is implicitly given in \cite{BMZ24} and turns out to be 
\begin{equation}\label{eq:theisometry}
    V_P\ket{\psi}_{AB} = \frac{1}{\sqrt{2^{3n}}}\sum_{a,b,c \in \{0,1\}^n} (-1)^{b \cdot c} X^a_P Z^b_P \ket{\psi}_{AB}\ket{c, a+c}
\end{equation}
for $P \in \{A, B\}$ where given strings $a=(a_1, \dots, a_n) , b=(b_1, \dots, b_n)$ we define
\begin{align*}
    X^a_P &\coloneqq X_P(a_1')^{a_1}X_P(a_2')^{a_2} \cdots X_P(a_n')^{a_n}, \\
    Z^b_P &\coloneqq Z_P(b_1')^{b_1}Z_P(b_2')^{b_2} \cdots Z_P(b_n')^{b_n},
\end{align*}
with $X_P(a_i'):=\tau^P(X(0^{i-1} a_i 0^{n-i})$, denoting provers' observable when receiving a query $X(0^{i-1} a_i 0^{n-i})$, and $Z_P(b_j'):=\tau^P(Z(0^{j-1} b_j 0^{n-j}))$, denoting their observable for question $Z(0^{j-1} b_j 0^{n-j})$. We remark that this isometry is efficiently implementable in a similar way to the swap gate. \footnote{Here by efficient we mean efficient to implement given oracle access to the players as outlined in \cref{def:OracleAccess}} We will revisit this isometry, in particular $V_B$, in \Cref{sec:implement-isometry}.

\subsection{Energy test and Hamiltonian game}

We describe the Energy test (Definition 21 of \cite{BMZ24}) in Figure \ref{fig:energytest}. The verifier randomly chooses a Hamiltonian term $H_\ell$. He requests Alice to teleport the ground state to Bob and simultaneously sends the corresponding tensor product of Pauli matrices to Bob. This allows the verifier to evaluate the energy of the ground state.

\begin{figure}[h]
\rule[1ex]{\textwidth}{0.5pt}
\begin{enumerate}[label=\arabic*.]
    \item The verifier picks $ \ell \in_R [m]$ at random for a local Hamiltonian $H_\ell$, and selects a pair at random uniformly from the set $D_\ell \coloneqq \{(W, e) \in \{X, Z\}^n \times \{0,1\}^n \mid \sigma_W(e) = H_\ell\}$.
    \item The verifier tells Alice that the players are under the energy test, and sends $W(e)$ to Bob.
    \item Alice answers with two $n$-bit strings $\alpha,\beta \in \{0,1\}^n$, and Bob with $c \in \{\pm1\}$.
    \item The verifier computes a $n$-bit string $d \in \{-1,+1\}^n$ as follows: Set $d_i = (-1)^{\alpha_i}$ if $W(e)_i = Z$, $d_i = (-1)^{\beta_i}$ if $W(e)_i = X$, and $d_i = 1$ otherwise.
    \item If $c \cdot \prod_{i \in [n]} d_i  \ne \textit{sign}(\gamma_\ell)$, the verifier accepts. 
    \item Otherwise, the verifier rejects with probability $|\gamma_\ell|$.
\end{enumerate}
\rule[2ex]{\textwidth}{0.5pt}\vspace{-.5cm}
\caption{Energy Test for a $XZ$ Hamiltonian $H = \sum_\ell \gamma_\ell H_\ell$.} 
\label{fig:energytest}
\end{figure}

Now, we can define the Hamiltonian game consisting of the LWPBT and the ET.
\begin{definition}[Hamiltonian game]\label{def:HamiltonianGame}
    Let $H = \sum_{\ell \in [m]} \gamma_\ell H_\ell$ be a $n$-local Hamiltonian of XZ-type and let $p \in (0, 1)$. The Hamiltonian game $\mathcal{G}(H, p)$ is a non-local game where the players play the low-weight PBT in \Cref{fig:pbt} with probability $(1-p)$, and the players play ET in \Cref{fig:energytest} with probability $p$.
\end{definition}

\begin{figure}
    \centering
    \begin{tikzpicture}[>=stealth, thick]

      \node[draw, rectangle] (Pa) at (0,0) {$P_A[w]$};
      \node[draw, rectangle] (Pb) at (4,0) {$P_B$};
      \node[draw, rectangle] (V) at (2,-3) {$V$};

      \draw[transform canvas={xshift=-13.pt}, <-, shorten <=5pt] (Pa) -- (V) node[midway, below left, align=center] {$q_1$};
      \draw[transform canvas={xshift=-3.pt}, ->, shorten >=5pt, shorten <=3pt] (Pa) -- (V) node[midway, right, align=center] {$c_1$};
      \draw[transform canvas={xshift=3.pt}, ->, shorten <=5pt, shorten >=5pt] (Pb) -- (V) node[midway, left, align=center] {$c_2$};
      \draw[transform canvas={xshift=13.pt}, <-, shorten <=5pt] (Pb) -- (V) node[midway, below right, align=center] {$q_2$};
      \draw[double distance=3pt, nfold=3, shorten >=5pt, shorten <=5pt]  (Pa) -- (Pb) node[midway, above] {EPR};
    \end{tikzpicture}
    \caption{Hamiltonian game $\mathcal{G}(H, p)$. The verifier plays the low-weight PBT with probability $1-p$ and ET with probability $p$. Depends on the verifier's choice of the test, Alice receives a query $q_1$, which is either a (tensor product of) Pauli matrices query or a teleportation query. However, Bob always (except one query) receives a (tensor product of) Pauli matrices query $q_2$. }
    \label{fig:HamiltonianGame}
\end{figure}

Because of the rigidity of the low-weight PBT, we can ascertain that the provers' measurement is close to the Pauli matrices (up to isometry) if the provers win the low-weight PBT with high probability. However, we do not know how Alice behaves in ET. Thus, we need to define a \emph{semi-honest} strategy $S_{sh}$. 
\begin{definition}
    We refer to a strategy $\mathcal{S}_{sh}$ for $\mathcal{G}(H,p)$ as a \emph{semi-honest}\footnote{We define the \emph{honest strategy} $\mathcal{S}_h$ for $G(H,p)$ in which the players employ the honest strategy when playing LWPBT, and in the energy test, Alice honestly teleports an unknown state $\eta$ to Bob and provides the verifier with the teleportation keys. } strategy if Alice and Bob hold $n$-EPR pairs and Bob must perform $\sigma_{W(a)}$ on question $W(a)$ since he cannot distinguish questions from low-weight PBT or ET.
\end{definition} 

\section{Extractor for the Hamiltonian game}\label{Sec:pok-extractor}
In this section, we define an extractor for the Hamiltonian game of \Cref{def:HamiltonianGame}. The intuition is as follows:  Under the assumption that the provers pass the low-weight PBT with large probability, we can ascertain by \Cref{thm:rigidity} that they are honest (up to isometry) in the sense that they share sufficiently many EPR pairs and perform the correct Pauli matrices.
To obtain a witness from the provers, the extractor first commands Alice to teleport her ground state to Bob. Then it exchanges its internal state with the second prover's state by performing a ``swap gadget'' constructed by Bob's observables in the game. As a consequence, we show that the extractor succeeds in extracting the prover's ground state.

We first explore the isometry in \Cref{eq:theisometry} to implement it efficiently in a quantum circuit in \Cref{sec:implement-isometry}. Then we finally show the existence of an extractor for $XZ$ local Hamiltonian problems in \Cref{sec:show-extractor}.

\begin{figure}
\centering
    \begin{quantikz}
    \lstick{}& & \gate{\sigma_X} & & \gate{\sigma_Z} & & \gate{\sigma_X} &   \\
    \lstick{\ket{0}} & & \ctrl{-1} & \gate{H} & \ctrl{-1} & \gate{H} &\ctrl{-1} &  \\
    \end{quantikz}
    =
    \begin{quantikz}
    \lstick{}& & \gate{\sigma_Z} & & \gate{\sigma_X} &   \\
    \lstick{\ket{0}}& \gate{H} & \ctrl{-1} & \gate{H} &\ctrl{-1} &  \\
    \end{quantikz}
\caption{A quantum circuit implementing the swap gate. The top wire represents the target register, while the bottom wire is an auxiliary register initialized in~$\ket{0}$.}
\label{fig:swapgate}
\end{figure}

\subsection{The isometry as a quantum circuit} \label{sec:implement-isometry}
To build an extractor for the Hamiltonian game, we make use of the isometry $V \coloneqq V_B$ in \Cref{eq:theisometry} acting on Bob's register as a gadget. We will see that the isometry  can be interpreted as a swap gate that essentially interchanges the states between Bob's register and the extractor's register. To implement this isometry, we first initialize the extractor's register with $n$-EPR pairs. For the simplicity, we denote the Bob's state and observable by $\ket{\psi}_B$.

\begin{figure}[h]
\centering
    \begin{quantikz}[wire types={q, q, q}, classical gap=0.03cm]
    \lstick{$\ket{\psi}_B$}& \gategroup[3,steps=5,style={rounded corners, inner xsep=2pt},background,label style={label position=below,anchor=north,yshift=-0.2cm}]{{\sc swap gadget: $U$}} & \gate{Z} & & \gate{X} &&\\
    \makeebit[]{$\ket{\Phi^+_n}_{E}$} & \gate{H^{\otimes n}} & \ctrl{-1} & \gate{H^{\otimes n}} & \ctrl{-1} & \targ{0}  & \\
    &&&&& \ctrl{-1} &
    \end{quantikz}
\caption{The circuit implementing the isometry $V$ of \Cref{eq:theisometry}. The state in Bob's register is denoted by $\ket{\psi}_B$.}
\label{fig:implement-isometry}
\end{figure}

The isometry of \Cref{eq:theisometry} can be implemented as follows (see \Cref{fig:implement-isometry}): First, the register is initialized with $\ket{\psi}_B\ket{\Phi^+}^{\otimes n} = \frac{1}{\sqrt{2^n}}\sum_{a \in \{0,1\}^n} \ket{\psi}_B\ket{a, a} = \frac{1}{\sqrt{2^n}}\sum_{a \in \{0,1\}^n} \ket{\psi}_B\ket{a, a}$. Applying a Hadamard gate yields
\[
    \frac{1}{2^n}\sum_{a,b \in \{0,1\}^n} (-1)^{b \cdot a} \ket{\psi}_B\ket{b, a}.
\]
Then the controlled-$Z$ applies to the first register to obtain
\[
    \frac{1}{2^n}\sum_{a,b \in \{0,1\}^n} (-1)^{b \cdot a} Z_B(b_1')^{b_1} \cdots Z_B(b_n')^{b_n}\ket{\psi}_B\ket{b, a},
\]
where each $Z_B(b_i')$ gate is performed controlled on each bit string $b_i$ of $\ket{b} = \ket{b_1 \cdots b_n}$.
Again, applying the Hadamard gate yields
\[
    \frac{1}{\sqrt{2^{3n}}}\sum_{a,b,c \in \{0,1\}^n} (-1)^{b \cdot (a+c)} Z_B(b_1')^{b_1} \cdots Z_B(b_n')^{b_n}\ket{\psi}_B\ket{c, a}.
\]
After the controlled-$X$ gate applies to the first register, we obtain
\begin{align*}
    &\frac{1}{\sqrt{2^{3n}}}\sum_{a,b,c \in \{0,1\}^n} (-1)^{b \cdot (a+c)} \Big(X_B(c_1')^{c_1}\cdots X_B(c_n')^{c_n}\Big)\Big(Z_B(b_1')^{b_1}\cdots Z_B(b_n')^{b_n}\Big)\ket{\psi}_B\ket{c, a} \\
    &\qquad = \frac{1}{\sqrt{2^{3n}}}\sum_{a,b,c \in \{0,1\}^n} (-1)^{b \cdot c} \Big(X_B(a_1')^{a_1}\cdots X_B(a_n')^{a_n}\Big)\Big(Z_B(b_1')^{b_1}\cdots Z_B(b_n')^{b_n}\Big)\ket{\psi}_B\ket{a, a+c}.
\end{align*}
Finally, the CNOT gate applied inside the extractor's register returns the state
\[
    \frac{1}{\sqrt{2^{3n}}}\sum_{a,b,c \in \{0,1\}^n} (-1)^{b \cdot c} \Big(X_B(a_1')^{a_1}\cdots X_B(a_n')^{a_n}\Big)\Big(Z_B(b_1')^{b_1}\cdots Z_B(b_n')^{b_n}\Big)\ket{\psi}_B\ket{c, a+c}.
\]
Therefore, if we let $U$ be the unitary implementing the isometry $V(=V_B)$ of \Cref{eq:theisometry} in the circuit, then we have
\begin{align}\label{eq:theunitary}
    U\ket{\psi}_B\ket{\Phi^+_n}
    &= \frac{1}{\sqrt{2^{3n}}}\sum_{a,b,c \in \{0,1\}^n} (-1)^{b \cdot c} X_B^aZ_B^b\ket{\psi}_B\ket{c, a+c}.
\end{align}

However, we can interpret the isometry (unitary) above as the swap gate. Indeed, let us consider the swap gate shown in \Cref{fig:enhanced-swapgate}, where the bottom two wires represent the auxiliary registers initialized with $n$-EPR pairs.  The swap gate $U_{\text{swap}}$ is represented by 
\begin{equation}\label{eq:theswapgate}
    U_{\text{swap}} \ket{\psi}_B\ket{\Phi^+_n} = \frac{1}{\sqrt{2^{3n}}}\sum_{a,b,c \in \{0,1\}^{n}}(-1)^{b \cdot c}\sigma_X(a)\sigma_Z(b)\ket{\psi}_B\ket{c, a+c}.
\end{equation}
If the provers behave honestly, that is, share the $n$-EPR pairs and perform the indicated Pauli measurement, then the isometry of \Cref{eq:theisometry} (or the unitary of \Cref{eq:theunitary}) obtained in \Cref{thm:rigidity}  coincides with the true swap gate of \Cref{eq:theswapgate}.

\begin{figure}[h]
\centering
    \begin{quantikz}[wire types={q, q, q}, classical gap=0.03cm]
    \lstick{$\ket{\psi}_B$}&  \gate{\sigma_X} && \gate{\sigma_Z} & & \gate{\sigma_X} &\\
    \makeebit[]{$\ket{\Phi^+_n}$} & \ctrl{-1} & \gate{H^{\otimes n}} & \ctrl{-1} & \gate{H^{\otimes n}} & \ctrl{-1} &  \\
    &&&&&&
    \end{quantikz}
\caption{The circuit implementing the swap gate $U_{\text{swap}}$. It has $n$-EPR pairs for the auxiliary register, which makes it robust to error in the amplitude of $\ket{0 \cdots 0}$ of $\ket{\psi}_B$ compared to the model of \Cref{fig:swapgate}~\cite{McK16}.}
\label{fig:enhanced-swapgate}
\end{figure}

\subsection{Extractor}\label{sec:show-extractor}
The rigidity of the low-weight PBT (for Bob) implies the possibility of using the isometry $V$ (acting on Bob's register) to extract the EPR pairs from Bob, and receive the state that is (supposedly) teleported by Alice. Moreover, in order to implement $V$ in  \Cref{eq:theisometry}, the extractor uses black-box access to Bob's observables.

In more details: Suppose that $\sigma_W(e) = H_\ell$  for some $W \in \{X, Z\}^n$ and $e \in \{0,1\}^n$. Let us denote by $M_{a,b}$ Alice's measurement with outcome $\{a,b\}$ as an answer to the teleportation query. Then, by post-selecting the state after teleportation, \Cref{thm:rigidity} implies that
\[
    \bra{\psi} M_{a,b} \otimes V^\dagger H_\ell V \ket{\psi}_{AB} 
        \approx
        \bra{\psi} M_{a,b} \otimes W_B(e) \ket{\psi}_{AB}.
\]
Remark that the left-hand side corresponds to the expectation over the honest Bob's answer conditioned on Alice's measurement $(a,b)$, while the right hand side corresponds to the expectation over Bob's answer in the actual protocol conditioned on Alice's measurement $(a,b)$. Then we claim that if the provers succeed in the Hamiltonian test with large probability, then Alice is actually teleporting to the extractor's register a sufficiently low energy state relative to the local Hamiltonian $H$. 

\begin{figure}[h]
\centering
   \begin{quantikz}[wire types={q, q, q, q, q, q}]
        \lstick[2]{$M_A$} &\gate[3]{M} \hphantom{very wide}&&& \ctrl[vertical wire=q]{4} && \rstick{$a$} \\
        \lstick{} &&&& & \ctrl[vertical wire=q]{3} & \rstick{$b$} \\
        \lstick[2]{$\ket{\psi}_{AB}$} &&  \\
        & \gate[3]{U} \hphantom{very wide} &\\
        \makeebit[]{$\ket{\Phi^+_n}_{E}$} &&&&  \gate{\sigma_X^a} & \gate{\sigma_Z^b} & \rstick{$\zeta$\\output state} \\
        &&
    \end{quantikz}
\caption{The model of our knowledge extractor. The top two wires represent the public message register that Alice reads. The middle two represent Alice's and Bob's register respectively.
The bottom two wires correspond to the extractor's register initialized in state $\ket{\Phi_n^+}_E$. Since the extractor can access to the controlled version of provers' observables coherently, it makes Bob perform a swap gadget $U$ of \Cref{fig:implement-isometry}, while making Alice perform quantum teleportation $M$. The extractor returns the output state $\zeta$ in the output register.}
\label{fig:extractor-circuit}
\end{figure}

We find the explicit form of the extractor's output state $\zeta$ in \Cref{fig:extractor-circuit}.
Let us set the circuit's initial state $\rho$ by
\[
    \rho = \ket{q}\bra{q} \otimes \ket{\psi}\bra{\psi}_{AB} \otimes\ket{\Phi^+_n}\bra{\Phi^+_n}_E,
\]
where state $\ket{q}$ denotes the classical message that Alice reads to perform the quantum teleportation, $\ket{\psi}_{AB}$ is the (purified) provers' sharing initial state, and $\ket{\Phi^+_n}_E$ is the $n$-EPR pairs prepared in the extractor's register.
Then state after implementing the unitary gate $U$ in \Cref{fig:extractor-circuit} becomes
\[
    (I_{\mathcal{M}_A} \otimes U)\rho (I_{\mathcal{M}_A} \otimes U)^\dagger = \ket{q}\bra{q} \otimes U(\ket{\psi}\bra{\psi}_{AB} \otimes\ket{\Phi^+_n}\bra{\Phi^+_n}_E)U^\dagger= \ket{q}\bra{q} \otimes \rho_1,
\]
where $\rho_1 = U(\ket{\psi}\bra{\psi}_{AB} \otimes\ket{\Phi^+_n}\bra{\Phi^+_n}_E)U^\dagger$.
Simultaneously, the first prover performs the quantum teleportation. Let $q_{a, b}$ be a non-negative number to denote a probability of the teleportation key outcome being a pair of $(a, b)$. Then we decompose the POVM $M$ into a family of mutually disjoint POVMs $\{ M_{a, b}\}_{a, b \in \{0,1\}^{n}}$, that is, $M = \sum_{a, b} \sqrt{q_{a, b}} M_{a, b}$. Note that each POVM $M_{a, b}$ collapses the state $\ket{q}\bra{q} \otimes \rho_1$ into $\ket{a, b}\bra{a, b} \otimes \rho_{a, b}$. 
Consequently, the state is transformed as follows.
Let us write 
\begin{align*}
    M_{a,b}\ket{q}\bra{q} \otimes \rho_1M_{c,d}^\dagger = 
    \begin{cases}
        \ket{a, b}\bra{b, a} \otimes \rho_{a,b}, & \text{ if } (a,b) = (c,d) \\
        \ket{a, b}\bra{d, c} \otimes \rho_{a,b, c,d} & \text{ if } (a,b) \neq (c,d).
    \end{cases}
\end{align*}
The state is transformed as follows.
\begin{align*}
    (M \otimes U)\rho (M \otimes U)^\dagger 
    &= M (\ket{q}\bra{q} \otimes \rho_1) M^\dagger \\
    &= \sum_{a, b.c, d} \sqrt{q_{a, b} q_{d, c}}\ket{a, b}\bra{c, d} \otimes \rho_{a,b,c,d}.
\end{align*}
At the end of the circuit, the extractor corrects the state in the output register by performing controlled Pauli matrices.
\[
    \sum_{a, b.c, d} \sqrt{q_{a, b} q_{c,d}}\ket{a, b}\bra{d, c} \otimes \sigma_Z^b\sigma_X^a\rho_{a,b,c,d}\sigma_X^c\sigma_Z^d.
\]
By tracing out to the output register, we have the output state $\zeta$ in the following form.
\begin{equation}\label{eq:outputstate}
    \zeta = \sum_{a, b} q_{a, b} \tr_{\overline{\text{out}}}( \sigma_Z^b \sigma_X^a \rho_{a, b} \sigma_X^a \sigma_Z^b)
    = \sum_{a, b} q_{a, b} \sigma_Z^b \sigma_X^a \rho_{a, b} \sigma_X^a \sigma_Z^b,
\end{equation}
where, in the second equality, we abuse the notation to identify $\sigma_Z^b \sigma_X^a \rho_{a, b} \sigma_X^a \sigma_Z^b$ with $\tr_{\overline{\text{out}}}( \sigma_Z^b \sigma_X^a \rho_{a, b} \sigma_X^a \sigma_Z^b)$. Overall, closer to honest the provers are (i.e. share the $n$-EPR pairs, and honestly perform the Bell measurement and Pauli matrices), closer to the ground state the extracted state becomes. 

Next, we estimate the quality of the extracted state $\zeta$ by relating it to the winning probability of the Hamiltonian game. In detail, we claim that the provers' winning probability can be upper-bounded by the winning probability when using $\zeta$ as a witness.

\begin{lemma}\label{lem:bmz-Upper-Bound Lemma}
Let $H=\frac{1}{m}\sum^m_{\ell=1} \gamma_\ell H_\ell$ be a $6$-local, $n$-qubit Hamiltonian of $XZ$-type. Assume that $\zeta$ is the output state of the extractor. For any $\eta\in(0,1)$, let $p=\frac{4\eta^{\nicefrac{3}{4}}}{3^{\nicefrac{3}{4}}(1+C)n^{6}}$ where $C$ is the constant given in \Cref{thm:rigidity}. Then
\begin{align*}
    \omega\big(\mathcal{G}(H,p) \big) \leq \omega(S_{sh}) + \eta,
\end{align*}
where $\omega(S_{sh})=1-p\left(\frac{1}{2m} \sum_{\ell\in[m]} |\gamma_\ell | +\frac{1}{2} \tr(H\zeta) \right)$.
\end{lemma}
\begin{proof}
    Suppose the provers are employing a strategy $\mathcal{S}=(\tau^A,\tau^B,\ket{\psi}_{AB})$ for $G(H,p)$ that wins LWPBT with probability $1-\delta_1$. Then they win the energy test with probability at most $\delta_2 +1-\frac{\sum_\ell |\gamma_\ell|}{2m} - \frac{\tr(H\zeta)}{2}$ with $\delta_2 \leq Cn^6\delta_1^{1/4}$ by \Cref{thm:rigidity} (See \Cref{sec:append} for details).
Then for $p \coloneqq \frac{4n^{-6}\eta^{\nicefrac{3}{4}}}{(1+C)3^{\nicefrac{3}{4}}}$ with $\eta\in(0,1)$, we have $\eta+\delta_1=\eta/3+\eta/3+\eta/3+\delta_1 \geq 4(\frac{\eta^3\delta_1}{3^3})^{\nicefrac{1}{4}}=p(1+C)n^{6}\delta_1^{\nicefrac{1}{4}}$. It follows that
\begin{align*}
    p\delta_2-(1-p)\delta_1 &\leq pCn^{6}\delta_1^{\nicefrac{1}{4}}+p\delta_1-\delta_1 \\
    &\leq pCn^{6}\delta_1^{\nicefrac{1}{4}}+pn^{6}\delta_1^{\nicefrac{1}{4}}-\delta_1
    = p(1+C)n^{6}\delta_1^{\nicefrac{1}{4}}-\delta_1 \\
    &\leq \eta.
\end{align*}
Hence, the overall winning probability is bounded as follows
\begin{align*}
    \omega(\mathcal{S}) &\leq(1-p)(1-\delta_1) + p(\delta_2 +1-\frac{\sum_\ell |\gamma_\ell|}{2m} - \frac{\tr(H\zeta)}{2}) \\
    &= \omega(\mathcal{S}_{sh}) +p \delta_2 -(1-p)\delta_1 \\
    &\leq \omega(\mathcal{S}_{sh})+\eta.
\end{align*}
\end{proof}

Our next lemma is required to obtain a knowledge function which does not depend on the energy gap $\beta-\alpha$ from an XZ Hamiltonian problem $(H,\alpha,\beta)$.

\begin{lemma}\label{lem:gamma}
    Let $H=\frac{1}{m} \sum_{\ell\in[m]} \gamma_\ell  H_l$ be an XZ-Hamiltonian  for which $\gamma := \frac{1}{m} \sum_{\ell\in[m]} |\gamma_\ell | \leq 1$. For any value  $0 \leq \alpha \leq \gamma$ there exists $\eta \in(0,1)$ such that for $p$ as defined from $\eta$ in \Cref{lem:bmz-Upper-Bound Lemma}
    \begin{equation*}
        1 - p^*(\frac{1}{2m}\sum|\gamma_\ell| + \frac{\alpha}{2}) = 1 - O(1/\text{poly}(n))
    \end{equation*}
\end{lemma}

\begin{proof}
Suppose that we can choose $\eta$ such that  the following relation holds.
\begin{equation}\label{eq:parameter-eta}
    \eta = \frac{p(\gamma + \alpha)}{2}.
\end{equation} 
In this case we have  $1 - p^*(\frac{1}{2m}\sum|\gamma_\ell| + \frac{\alpha}{2}) = 1 - \eta$. To achieve this we need a choice of $\eta$ which satisfied both \Cref{eq:parameter-eta} as well as $p=\frac{4\eta^{\nicefrac{3}{4}}}{3^{\nicefrac{3}{4}}(1+C)n^{6}}$. The explicit formulas for parameters $\eta$ and $p$ which satisfy both equations, denoted by $\eta^*$ and $p^*$ respectively, are:
\begin{align}
    \eta^* = \frac{16(\gamma + \alpha)^4}{27(1+C)^4n^{24}},\label{eq:eta} \\
    p^* = \frac{32(\gamma + \alpha)^3}{27(1+C)^4n^{24}}. \label{eq:p}
\end{align}
Note, since $\gamma \leq 1 < C$ we have $\eta^* <1$ for all values of $n$.
\end{proof}

\begin{figure}[H]
\rule[1ex]{\textwidth}{0.5pt}
Let $H$ be a XZ local Hamiltonian, and $P^* = (P_A, P_B)$ be a non-communicating two-prover with a strategy $\mathcal{S}$ such that $\omega(\mathcal{G}, \mathcal{S}) \geq 1 - \epsilon$ for the Hamiltonian game $\mathcal{G} = \mathcal{G}(H, p^*)$ (\Cref{def:HamiltonianGame}).
\begin{enumerate}[label=\arabic*]
  \item Initialize the state in the circuit to $\ket{\psi}_{AB}\ket{\Phi^+_n}_E$.
  \item Request Alice to perform the quantum teleportation by sending the classical message $\ket{q}$ through the message register $M_A$.
  \item Simultaneously, implement the swap gadget $U$ of \Cref{eq:theunitary} by black-box access to Bob's observables.
  \item Correct the state in the extractor's output register by applying controlled Pauli matrices $\sigma_X$ and $\sigma_Z$, where the message register $M_A$ acts as a control, to obtain $\zeta$.
  \item Output $\zeta$.
\end{enumerate}
\rule[2ex]{\textwidth}{0.5pt}\vspace{-.5cm}
\caption{Knowledge extractor $E$.}
\label{fig:extractor}
\end{figure}

\begin{theorem}\label{prop:main}
 Let $H=\frac{1}{m} \sum_{\ell\in[m]} \gamma_\ell  H_l$ be an XZ-Hamiltonian  for which $\gamma := \frac{1}{m} \sum_{\ell\in[m]} |\gamma_\ell | \leq 1$. For any value $0 \leq \alpha \leq \gamma$ take $p^*$ and $\eta^*$ to be as defined in \cref{eq:eta} and \cref{eq:p}. If provers $P^*$, with some strategy $\mathcal{S}$, wins $\mathcal{G}(H,p^*)$ with probability $1 - \epsilon > 1-\eta^*$, then the extractor of \Cref{fig:extractor},  with oracle access to $\mathcal{S}$, outputs a quantum state $\zeta$ such that $\tr(H\zeta) \leq \alpha + O(\poly(n))\,\epsilon.$
\end{theorem}

\begin{proof}
Suppose that the provers $P^*$ with some strategy $\mathcal{S}$ wins the Hamiltonian game with probability $\omega\big(\mathcal{G}(H,p), \mathcal{S} \big) = 1 - \epsilon \geq 1 - \eta^*$. By \Cref{lem:bmz-Upper-Bound Lemma} implies that the extractor $E^\mathcal{S}$ outputs the state $\zeta$ such that 
 \begin{equation}
     \quad  1 - \epsilon \leq 1-\frac{p^*}{2}\Big(\gamma + \tr(H\zeta) \Big) + \eta^* \quad (\leq 1).
 \end{equation}
It follows from a choice of parameter $\eta^*$ in \Cref{eq:parameter-eta} that we obtain
\begin{equation}
    \tr(H\zeta) \leq \left(\frac{2\eta^*}{p^*} - \gamma\right) + \frac{2 \epsilon}{p^*} = \alpha + \frac{2\epsilon}{p^*} = \alpha + Dn^{24}\epsilon,
\end{equation}
where $D = \frac{27(1+C)^4}{16(\gamma + \alpha)^3}$.

\end{proof}

\begin{remark}
We note that \Cref{prop:main} does not apply for $XZ$ local Hamiltonians instances $x=(H, \alpha, \beta)$ with $\alpha <0$, such as when $H=-I$ and $\alpha=\lambda_0(H)=-1$. For our purposes it suffices to consider the implications of $\cref{prop:main}$ for certain $\QMA$ relations $R_{Q, \gamma}$. More specifically, for any instance $x$ in a $\QMA$ language $A=(A_{yes}, A_{no})$ we may pick the verification circuit $V_x$ to be as specified by Theorem 27 of ~\cite{BMZ24}. Applying the circuit to Hamiltonian construction to $V_x$ will then give an $XZ$ local Hamiltonian instance $(H_x, \alpha, \beta)$ with $0 \leq \alpha < \beta$. For such instances we can obtain a knowledge error $\kappa$ of the form $\kappa=1 - O(1/\text{poly}(n))
$. 
\end{remark}

\begin{theorem}\label{thm:main}
Consider an instance $x$ of a $\QMA$ problem $A = (A_{\text{yes}}, A_{\text{no}})$, and define $(H_x, \alpha, \beta)$ as the XZ Hamiltonian obtained by applying the circuit-to-local Hamiltonian construction to the verification circuit $V_x$ of \cite{BMZ24}. Additionally, let $Q$ denote the standard $\QMA$ verification circuit for  $(H_x, \alpha, \beta)$. There exists polynomials $q, r$ such that game $\mathcal{G}(H, p^*)$ is a statistical zero-knowledge classical proof of quantum knowledge for QMA relation $R_{Q, 1-q(n, \epsilon)}$ with knowledge error $\kappa=1-\frac{1}{r(n)}$. 
\end{theorem}

\begin{proof}
We define our knowledge error as $\kappa :=1 - \hat{\eta}$, where
$\hat{\eta}:=\frac{16\gamma^4}{27(1+C)^4n^{24}}$. Note that since $\alpha \geq 0$ we have $\hat{\eta} \leq \eta^*$. Suppose provers $P^*$, with some strategy $\mathcal{S}$, wins $\mathcal{G}(H,p^*)$ with probability  $1 - \epsilon \geq \kappa$. Then by \cref{prop:main} the output state, $\zeta$, of the extractor satisfies $\tr(H\zeta) \leq  \alpha + Dn^{24}\epsilon$, where $D= \frac{27(1+C)^4}{16(\gamma + \alpha)^3}$. Such a state will be accepted by $Q$ with probability at least $1-\frac{\alpha + Dn^{24}\epsilon}{m}$.

To establish statistical zero-knowledge, we observe that the game $\mathcal{G}(H, p^*)$ differs from the Hamiltonian game of~\cite{BMZ24} only in its question distribution. The simulator constructed in~\cite{BMZ24} is designed to handle malicious verifiers that sample questions from any efficiently sampleable distribution. Therefore, it applies directly to $\mathcal{G}(H, p^*)$ without modification.
\end{proof}

\bibliographystyle{bibtex/bst/alphaarxiv.bst}
\bibliography{bibtex/bib/quasar-full.bib,
              bibtex/bib/quasar.bib,
              bibtex/bib/quasar-more.bib}
\appendix
\section{Upper bound of the winning probability}\label{sec:append}
Recall that from \Cref{eq:outputstate} we have the output state of the extractor
\begin{equation*}
    \zeta 
    = \sum_{a, b} q_{a, b} (\sigma_Z^b \sigma_X^a \rho_{a, b} \sigma_X^a \sigma_Z^b)
    = \sum_{a, b} q_{a, b} \zeta_{a, b},
\end{equation*}
where we define $\zeta_{a, b} = \sigma_Z^b \sigma_X^a \rho_{a, b} \sigma_X^a \sigma_Z^b$. We want to estimate the energy of this state. To do so, we relate it to the winning probability of the Hamiltonian game.
Let $H_\ell$ be the $\ell$th  term of the given local Hamiltonian $H$. Consider the semi-honest strategy $\mathcal{S}_{sh}$ of the provers (this coincides with the honest provers with witness $\zeta$). Since it can be seen that Bob performs the correct Pauli measurements with respect to state $\zeta$, we may assume that Bob has witness $\zeta$ when ET being conducted. When he receives the query $W(e) = H_\ell$, then the expectation over his respond $c$ is given by 
\begin{align*}
    \tr(H_\ell \zeta) = \sum_{a,b} q_{a,b}\tr(H_\ell \sigma_Z^b \sigma_X^a \rho_{a, b} \sigma_X^a \sigma_Z^b ) = \sum_{a,b} q_{a,b}d(a, b, H_\ell) \tr(H_\ell \rho_{a, b}),
\end{align*}
where $d(a, b, H_\ell)$ is the teleportation correction term as the function of $a, b$, and $H_\ell$ to return either $+1$ or $-1$. 

On the other hand, \Cref{thm:rigidity} states that the above equation is $O(n^6)\epsilon^{\nicefrac{1}{4}}$-close to $\tr((M_{a,b} \otimes H_\ell) \rho)$, which is the expectation over Bob's respond in the actual protocol conditioned on the outcome of measurement being $(a, b)$. Indeed, 
\begin{align*}
    \tr(H_\ell \rho_{a, b}) 
    &= \tr(H_\ell \tr_{A}((M_{a, b} \otimes V) \rho (M_{a, b} \otimes V)^\dagger)) \\
    &= \tr((I_A \otimes H_\ell) (M_{a, b} \otimes V) \rho (M_{a, b} \otimes V)^\dagger)\\
    &= \tr\left((M_{a, b} \otimes V^\dagger H_\ell V) \rho \right) \\
    &\approx_{O(n^6)\epsilon^{\nicefrac{1}{4}}} \tr((M_{a, b} \otimes W_B(e) \rho )
\end{align*}
by applying the Cauchy–Schwarz inequality and \Cref{thm:rigidity}.
Therefore, we have
\begin{equation}\label{eq:partialrigid}
    | \tr(H_\ell \zeta) - \sum_{a, b} q_{a, b} d(a, b, H_\ell)\tr(M_{a, b} \otimes W_B(e) \rho ) | \leq O(n^6)\epsilon^{\nicefrac{1}{4}}.
\end{equation}
However, the second term in the left hand side is the expectation over the transcripts in the actual protocol. Thus, the claim follows.

Let us denote by $p_{\text{loss}}^{(\et)}$  the losing probability in ET in the semi-honest strategy (or the honest strategy with witness $\zeta$). Similarly, we denote by $p_{\text{loss, act}}^{(\et)}$ the losing probability in ET of the actual provers. Recall that the losing probability in ET is given by
\begin{equation}
    \mathbb{E}_{\ell \leftarrow [m], H_\ell \leftarrow D_\ell, (a, b) \leftarrow K}\frac{|\gamma_\ell| + \gamma_\ell \mathbb{E}_{c \mid a, b} [c \cdot d]}{2},
\end{equation}
where $K$ denotes the distribution of the measurement outcome $(a, b)$, and so $\pr[(a, b) \leftarrow K] = q_{a, b} \geq 0$. We remark that if the provers behave semi-honestly, then we have $\mathbb{E}_{c \mid a, b} [c \cdot d] =  \tr(H_\ell \sigma_Z^b \sigma_X^a \zeta_{a, b} \sigma_X^a \sigma_Z^b )$. Therefore, the semi-honest prover fails ET with probability
\begin{align*}
     p_{\text{loss}}^{(\et)} 
     &= \mathbb{E}_{\ell \leftarrow [m], H_\ell \leftarrow D_\ell}\frac{|\gamma_\ell|}{2} + \frac{\gamma_\ell \sum_{a, b} q_{a, b}\tr(H_\ell \sigma_Z^b \sigma_X^a \zeta_{a, b} \sigma_X^a \sigma_Z^b)}{2} \\
     &= \mathbb{E}_{\ell \leftarrow [m], H_\ell \leftarrow D_\ell}\frac{|\gamma_\ell|}{2} + \frac{\gamma_\ell \tr(H_\ell \zeta)}{2} \\
     &\approx_{O(n^6)\epsilon^{\nicefrac{1}{4}}} \mathbb{E}_{\ell \leftarrow [m], H_\ell \leftarrow D_\ell} \left(\frac{|\gamma_\ell|}{2} + \frac{\gamma_\ell \sum_{a, b} q_{a, b} d(a, b, H_\ell)\tr((M_{a, b} \otimes W_B(e)) \rho )}{2} \right), \text{ where } W(e)=H_\ell\\
     &= \mathbb{E}_{\ell \leftarrow [m], H_\ell \leftarrow D_\ell, (a, b) \leftarrow K} \left(\frac{|\gamma_\ell|}{2} + \frac{\gamma_\ell \mathbb{E}_{d \mid a, b} [c \cdot d]}{2} \right) \\
     &= p_{\text{loss, act}}^{(\et)} ,  \\
\end{align*}
where the approximation in the third line follows from \Cref{eq:partialrigid}.
Therefore, the winning probability of the actual prover is at most $\omega(\mathcal{S}_{sh}) + O(n^6)\epsilon^{\nicefrac{1}{4}}$, where 
\[
    \omega(\mathcal{S}_{sh}) = 1-p\left(\frac{1}{2m} \sum_{\ell\in[m]} |\gamma_\ell | +\frac{1}{2} \tr(H\zeta) \right).
\] 
\end{document}